\documentclass[11pt, helvetica]{article}
  
\usepackage{verbatim}
\usepackage{times}
\usepackage{amsmath, amsthm, amssymb}
\usepackage{color}
\usepackage{graphicx}
\usepackage{cite}
\usepackage{epsfig}
\usepackage{url}
\usepackage{xcolor}
\usepackage{xspace}
\usepackage{pgfgantt}
\usepackage{booktabs}
\usepackage{authblk}

\newtheorem{theorem}{Theorem}

\newtheorem{lemma}{Lemma}

\newtheorem{definition}{Definition}

\newcommand{\Alg}{\textsc{RobustTalk}}
\newcommand{\LN}{\textsc{Iteration-1}}
\newcommand{\MLN}{\textsc{Iteration-2}}
\newcommand{\AT}{\textsc{Iteration-3}}
\newcommand{\MAT}{\textsc{Iteration-4}}

\newcommand{\AlgN}{\textsc{RobustBroadcast}}

\newif\ifcomments
\commentstrue

\newcommand{\Total}{T}
\def\myval{0pt}
\def\pval{0pt}
\def\nval{0pt}
\def\bigval{0pt}

\newcommand{\defn}[1]{\textbf{\emph{#1}}}

\begin{document}

\title{A Resource-Competitive Jamming Defense\footnote{This research was supported in part by an NSERC Discovery Grant, and NSF grants CCF-1217338, CNS-1318294, CCF-1514383, CCF-1637546, NSF CAREER Award 0644058, CCR-0313160, and CCF-1613772.}}
\date{}

\author[1]{Valerie King}
\author[2]{Seth Pettie}
\author[3]{Jared Saia}
\author[4]{Maxwell Young}
\affil[1]{\small Dept. of Computer Science, University of Victoria, BC, Canada \hspace{1cm}\texttt{val@cs.uvic.ca}}
\affil[2]{\small Dept. of Electrical Engineering and Computer Science,University of Michigan, MI, USA \texttt{pettie@umich.edu}}
\affil[3]{\small Dept. of Computer Science, University of New Mexico, NM, USA \mbox{\texttt{saia@cs.unm.edu}}}
\affil[4]{\small Computer Science and Engineering Dept., Mississippi State University, MS, USA Email: \texttt{myoung@cse.msstate.edu}}

\maketitle
\vspace{-1cm}
\begin{abstract}
Consider a scenario where  Alice wishes to send a message $m$ to Bob in a time-slotted wireless network. However, there exists an adversary, Carol, who aims to prevent the transmission of $m$ by jamming the communication channel. There is a per-slot cost of $1$  to send, receive or jam $m$ on the channel, and we are interested in how much Alice and Bob need to spend relative to Carol in order to guarantee communication.

Our approach is to design an algorithm in the framework of \defn{resource-competitive analysis} where the cost to correct network devices (i.e., Alice and Bob) is parameterized by the cost to faulty devices (i.e., Carol). We present an algorithm that guarantees the successful transmission of $m$ and has the following property: if Carol incurs a cost of $\Total$ to jam, then both Alice and Bob have a cost of $O(\Total^{\varphi - 1} + 1)=O(\Total^{.62}+1)$ in expectation, where $\varphi = (1+ \sqrt{5})/2$ is the golden ratio.  In other words,  it possible for Alice and Bob to communicate while incurring asymptotically less cost than Carol. We generalize to the case where Alice wishes to send $m$  to $n$ receivers, and we achieve a similar result.

Our findings hold  even if (1) $\Total$ is {\it unknown} to  either party; (2) Carol knows the algorithms of both parties, but not their random bits;  (3) Carol can jam using knowledge of past actions of both parties; and (4) Carol can jam reactively,  so long as there is sufficient network traffic in addition to $m$. 
\end{abstract}

\clearpage

\section{Introduction}\label{section:intro}

The wireless communication medium is vulnerable to \defn{jamming} whereby an attacker causes interference on the communication channel.   Such attacks are prohibited by US federal law, however, they occur in practice (see~\cite{jamming:marriott,jamming:redline}) and will  likely increase in frequency with the popularity of wireless devices.

Defending against jamming is particularly important to the security of  wireless low-power networks (LPNs) such as wireless sensor networks (WSNs) and the Internet of Things (IoT), where these attacks are known to degrade performance~\cite{6736871,xu:jamming}. Such networks are especially vulnerable given that  their constituent devices are often battery-powered and, therefore,  energy is a scarce resource that must be conserved. 

With this setting in mind, we consider the following problem: Alice wishes to guarantee transmission of a message $m$ to Bob over a wireless communication channel. However, there exists an adversary, Carol, who aims to prevent communication by jamming transmissions over the channel. Our solution to this problem makes use of a relatively new technique called \defn{resource competitiveness}. Informally, the idea behind resource-competitive analysis is to include the cost of the adversary, $T$, as another parameter by which to measure the algorithmic performance (details are provided in Section~\ref{subsection:definition}).   

We assume a time-slotted model where it costs $1$ to send, listen, or jam per slot on the communication channel; these operations typically dominate the operational costs of wireless devices in terms of energy usage (see Section~\ref{section:discussion} for discussion of this issue). Our results guarantee the successful transmission of $m$ and has the following costs: if Carol incurs a cost of $\Total$ to jam, then both Alice and Bob have a cost of $O(\Total^{\varphi - 1} + 1)=O(\Total^{.62}+1)$ slots in expectation, where $\varphi = (1+ \sqrt{5})/2$ is the golden ratio. We are able to generalize to the setting where Alice wishes to send $m$ to $n$ receivers, and we obtain a similar result. 

In both cases,  the smaller asymptotic costs to the good devices implies that it is possible to ``bankrupt'' a jamming adversary. That is, in attempting to prevent communication through jamming, Carol will deplete her onboard power supply far more rapidly than either Alice or Bob (or multiple receivers) and, when this occurs, communication will succeed. Note that such a result is useful even in the face of continuous jamming; this is in contrast to most prior work where typically there are constraints on how the adversary is permitted to jam (Section~\ref{sec:related_work} discusses previous work on jamming mitigation).

Our results holds even if (1) $\Total$ is {\it unknown}; (2) Carol knows the algorithms used by the communicating parties, but not their random bits; and (3) Carol can  attack using total knowledge of past actions of the communicating parties (i.e. Carol is an \defn{adaptive} adversary).  Under the assumption of sufficient network traffic (in addition to $m$), we also demonstrate that the result holds when Carol knows the channel is in use and then decides to jam (i.e. Carol is a \defn{reactive} adversary).


\subsection{Resource Competitiveness} \label{subsection:definition}

Resource competitiveness is a new approach to designing robust algorithms for distributed systems~\cite{gilbert:resource,bender:resource}. We now define what it means for a distributed algorithm $\mathcal{A}$ to be {\it resource competitive}.

Assume a system with a set $G$ of $n$ good nodes that obey actions prescribed by $\mathcal{A}$. There exists an adversary who incarnates a source of disruption in the system. For example, the adversary may represent (1) any number of malicious nodes that collude and deviate arbitrarily from $\mathcal{A}$, or (2) the effects of more benign failures due to software or hardware faults.

Let $Cost(\alpha, v)$ denote the resource expenditure (or cost) to a good node $v$ over an execution $\alpha$. A resource might be bandwidth, CPU cycles, energy, actual money, or another useful domain-specific measure.  $Cost(\alpha, v)$ is the cost incurred by $v$ for executing the actions prescribed by $\mathcal{A}$ in an execution $\alpha$. 

Let $T(\alpha)$ be the adversary's total cost; this is typically unknown to the good nodes. We simply refer to $T$ and $Cost(v)$  as $\alpha$ is implicit. We  define what it means for $\mathcal{A}$ to be resource competitive:\smallskip 

\begin{definition}\label{def:single} Algorithm $\mathcal{A}$ is ($\rho$, $\tau$)-resource competitive if $\max_{\mbox{\footnotesize $v\in G$}}$ $\{Cost(v)\} \leq   \rho(T) + \tau$ for $\tau>0$ and any execution.\smallskip  \end{definition} 

\noindent In other words,  $\mathcal{A}$ is resource competitive if the {\it maximum} cost incurred by any good node is less than some function of the adversary's total cost, $\rho(T)$, plus some additive term $\tau$. The function $\rho$ is called the {\it robustness function} and it is a function of $T$ and possibly other parameters such as $n$. In designing $\mathcal{A}$, we desire $\rho$  to be a slow-growing function. This implies that, as the amount of disruption increases, the cost incurred by $\mathcal{A}$ increases slowly. Indeed, in dealing with jamming attacks, we show that it is possible to achieve $\rho(T) = o(T)$.

Why do we require $\tau$?  Note that when $T=0$ --- that is, there is no disruption --- the good nodes clearly cannot incur less than zero cost; $\tau$ represents the unavoidable cost to attain a goal even in the absence of disruption.  Efficiency when the level of disruption is low, or non-existent, is critical to efficiency.  It useful to make this separation explicit via  $\tau$, which we call the {\it efficiency function}, and it can be a function of parameters such as $n$, but it is {\it not} a function of $T$.

Functions other than the maximum cost over all nodes may be appropriate depending on the context. For instance, we may consider the average or median cost.  However, Definition~\ref{def:single} has been applied in the existing literature. Finally, we typically report results using big-O notation of the form $O(\vspace{1pt}\rho(T) + \tau)$.


\subsection{Resource Competitiveness in Context}\label{section:resource_competitiveness_intro}

\noindent It is helpful to contrast the notion of resource competitive algorithms with a number of related concepts. \smallskip

\noindent{\bf Notions of Competitiveness.} Competitive analysis is a well-known technique that  evaluates the worst-case performance of an online algorithm relative to an optimal offline algorithm $\mathcal{OPT}$~\cite{sleator:amortized}. While the inputs to an online algorithm can be viewed as adversarially selected, there is no notion of cost to the adversary for selecting certain inputs over others. In contrast, resource competitiveness places the cost to the adversary directly in the performance metric (see Definition~\ref{def:single}). Unlike online analysis where it is impossible for an online algorithm to outperform $\mathcal{OPT}$, a resource-competitive algorithm can actually outperform the adversary by achieving $\rho(T) = o(T)$.  

Game theory provides another measure of competitiveness known as the  ``price of anarchy'' which is the ratio of the worst-case Nash equilibrium to the global social optimum~\cite{robinson:price,nisan:gametheory}. In resource-competitive analysis, each node either obeys the protocol or it does not; in game theory, nodes seek to maximize their respective utility functions. It is possible to address malicious behavior in the context of game theory (see~\cite{abraham:distributed,clement:theory,aiyer:bar,vilaca:asynchrony,li:bar}). The incorporation of  game theoretic concepts may be an interesting direction for {\it future} work on resource-competitive algorithms.\smallskip\smallskip

\noindent{\bf Notions of Inflicting Cost.} The idea of inflicting cost on an opponent arises more explicitly in the domain of cryptography. A  major differentiating aspect of  cryptographic approaches is that the length of a private key is decided prior to encryption. This roughly determines (i) how much the adversary must spend in order to compromise the cryptosystem, and (ii) how much the good nodes must spend to achieve a particular level of security. In contrast, resource-competitive algorithms are {\it adaptive} in the following sense. When $T=0$, there is a small upfront cost quantified by the efficiency function $\tau$. Then, as $T$ grows, the cost function $\rho(T)$ grows commensurately and will dominate the cost of the algorithm when $T$ grows large enough; in this sense, resource-competitive algorithms are adaptive. This is very different from having a predetermined cost.\smallskip\smallskip


\noindent{\bf Additional Related Ideas.} An early example of considering an attacker's resources is the cryptosystem by Merkle~\cite{merkle:secure} where computational puzzles are used to inflict cost on an eavesdropper. However, the hardness of the puzzles must be set {\it a priori} and, therefore this scheme lacks the adaptivity of resource-competitive algorithms described above. Inflicting computational cost has been used to deter spam email~\cite{dwork:pricing}. Another example arises in settings where an attacker controls multiple identities. In such a network, one may issue a cost for joining via computational puzzles~\cite{tegeler:sybil,li:sybilcontrol} or monetary penalties~\cite{castro:secure}. Similar ideas arise in proposals to mitigate applicaton-level DDoS attacks. Typically,  a client must spend bandwidth or computational resources prior to receiving service (see~\cite{walfish:ddos,liu:netfence,parno:portcullis}).  In contrast to these results, resource-competitive analysis goes beyond simply inflicting cost by quantifying a relationship between the cost to the adversary and the cost to the good nodes. 


\subsection{Related Work on Tolerating Jamming Attacks}\label{sec:related_work}\vspace{\nval}

Several works address applied security considerations with respect to jamming~\cite{bayraktaroglu:performance,xu:feasibility,lin:link,aschenbruck:simulative} where the adversary deliberately disrupts the communication medium. Defenses include spread spectrum techniques (frequency or channel hopping), mapping with rerouting, and others (see \cite{wood:denial,xu:jamming,spread:foiling,navda:using}). 

Recent applied work by Ashraf~{\it et al.}~\cite{ashraf:bankrupting} investigates a similar line of reasoning by examining ways in which multi-block payloads (each block in a packet has its own cyclic redundancy code), so-called ``look-alike packets'', and randomized wakeup times for receivers can be used to force the adversary into expending more energy in order to jam transmissions. The authors' approach is interesting and the use of look-alike packets to prevent an adversary from being able to differentiate between different types of nework traffic is similar to our approach to foiling reactive jammers (see Section~\ref{section:reactive_adv}). However, on the whole, their approach is quite different from our own; moreover, analytical results are not provided in~\cite{ashraf:bankrupting}.

There are a number of theoretical results on jamming adversaries. Gilbert and Young~\cite{gilbert:making} give a Monte Carlo $1$-to-$n$ partial broadcast algorithm that is resource competitive. However, the result critically depends on knowing $n$ and still allows the adversary to prevent a small, but constant, fraction of the nodes from receiving the broadcast. Removing the reliance on $n$ is challenging, but a Monte Carlo $1$-to-$n$ broadcast algorithm is given by Gilbert~{et al.}~\cite{gilbert:near} for this case. 

Similarly, the problem of interactive communication on noisy channels has been examined from a resource-competitive perspective. Results in Dani~\cite{dani:errorcorrection} and Dani~{\it et al.}~\cite{daniICJournal17,ICALP15} address an adversary that can flip an unknown number of bits as information is transmitted across a channel. However, this work differs critically in that the goal is to minimize the latency of communication rather than energy expenditure 

Recently, Bender~{\it et al.}~\cite{bender:how} demonstrated that contention resolution -- how to coordinate access to a shared channel by multiple devices -- can be accomplished despite a powerful jamming adversary. Informally, the result guarantees expected constant throughput in a setting where $n$ requests are scheduled adversarially, and each client is shown to require at most an expected polylogarithmic number of channel-access attempts


Outside of resource-competitive results, there is large body of work on mitigating jamming attacks in wireless sensor networks (see~\cite{young:overcoming} for a survey). Gilbert~{\it et al.}~\cite{gilbert:malicious} examines the duration for which communication between two parties can be disrupted in a model with collision detection in a time-slotted network against an adversary who interferes with an unknown number of transmissions. As we do here, the authors assume channel traffic is always detectable at the receiving end (i.e. silence cannot be ``forged''). The authors employ the notion of {\it jamming gain} which is, roughly speaking, the ratio of the duration of the disruption to the adversary's cost for causing such disruption. This is an interesting metric which can gauge the efficiency of the adversary's attack; however, it does not incorporate the cost incurred by the correct parties in the system. 

Pelc and Peleg~\cite{pelc:feasibility} examine an adversary that randomly corrupts messages. Results by Awerbuch~{et al.}~\cite{awerbuch:jamming} and Richa~{et al.}\cite{richa:jamming2,richa:jamming3,richa:jamming4} consider an
adversary whose jamming is bounded within any sufficiently large time window. Under this adversarial model, Ogierman~et al.~\cite{ogierman:competitive} study medium access in the signal-to-interference-plus-noise ratio communication model. In the context of Sybil attacks~\cite{douceur02sybil}, Gilbert and Zheng~\cite{gilbert:sybilcast} devise methods for tolerating a  jamming adversary using coordination amongst the communicating parties, and this is generalized to a completely decentralized scenario by Gilbert, Newport and Zheng in~\cite{gilbert:who}.

In settings where devices have access to more than one channel, a number of theoretical results have been proposed. Dolev~{\it et al.}~\cite{dolev:gossiping} address a variant of the gossiping problem when multiple  channels are jammed. Gilbert~{\it et al.}~\cite{gilbert:interference} derive bounds on the time required for information exchange when a reactive adversary jams multiple channels.
Meier~{\it et al.}~\cite{meier:speed} examine the delay introduced by a jamming adversary for the problem of node discovery, again in a multi-channel setting. Dolev~{\it et al.}~\cite{dolev:secure} address  secure communication using multiple channels with a non-reactive adversary.  Recently, Dolev~{\it et al.}~\cite{dolev:synchronization} consider wireless synchronization in the presence of a jamming adversary. Emek and Wattenhofer~\cite{emek:frequency}  examine the effectiveness of frequency hopping against an adversary who jams a strict subset of the available channels. 

There are also game-theoretic treatments for jamming attacks, and we refer the interested reader to the survey by Manshaei~{\it et al.}~\cite{manshaei:game} for a comprehensive treatment of this area. 

The problem of robust communication has also been considered in the context of reliable broadcast where devices are laid out on a grid model~\cite{king:sleeping,bertier:reliable,vaikuntanathan:reliable,king:sleeping_journal,bhandari,bhandari2,koo,bhandari:probabilistic}.
Listening costs are accounted for by King~{\it et al.}~\cite{king:sleeping,king:sleeping_journal}, but jamming adversaries are not considered. Alistarh~{\it et al.}~\cite{alistarh:securing} assume collision detection and achieve authenticated reliable broadcast with the use of cryptography. With a reactive jamming adversary, Bhandhari~{\it et al.}~\cite{koo2} give a reliable broadcast protocol
when the amount of jamming is bounded and known {\it a priori}; however, correct nodes must expend considerably more energy than the adversary. Progress towards fewer broadcasts is made by Bertier~{\it et al.}~\cite{bertier:icdcs}; however, each node spends significant time in the costly listening state.
 
Finally, a preliminary version of our results  appeared in~\cite{king:conflict}; however, we focus strictly on the single-hop case as this provides the most compelling  results. This current version of our work contains additional exposition regarding the design of algorithm design, revised (and, we believe, clearer) versions of certain proofs, along with arguments and discussion that were previously omitted. Additionally, we consider an overlooked attack that was not addressed in~\cite{king:conflict} (discussed in Section~\ref{sec:variance}). Finally, we note that the lower bound originally provided in~\cite{king:conflict} is not provided here, since more recent work in~\cite{gilbert:near} gives a stronger bound that is asymptotically tight.


\subsection{Our Model}\label{section:attack_model}

We describe our network model and define the communication problem addressed in this work.\vspace{\myval}

\noindent{\bf Las Vegas Property:} Communication of $m$ from Alice to Bob must be guaranteed with probability 1; that is, we require a Las Vegas algorithm. An obvious motivation for this Las Vegas property is a critical application,  such as the dissemination of an important security update, where success is paramount.

The Las Vegas property has additional merit in multi-hop wireless networks where Monte Carlo algorithms may not be adequate. In particular, the failure probability in any single hop of a route may be small in the size of the broadcast neighborhood $n$. However, for large networks of total size $N$, it may be the case that $n\ll N$ and so a union bound over the path length may fail to offer a high-probability guarantee of correctness.\smallskip

\noindent{\bf Channel  Utilization:} Sending or listening on the communication channel by Alice and Bob occurs in discrete units  called {\it slots}. For example, under the common IEEE 802.11g, a slot may correspond to an actual time slot ($9\mu$s) in a time division multiple access (TDMA) type access control protocol.  For simplicity, we assume that the message $m$ fits within a single slot; otherwise, we can send $m$ piecewise.

The cost for sending or listening is $1$ per slot. This is a normalized cost meant to reflect the fact that sending and listening on the channel typically dominates the operational costs of LPN devices. For example, in the WSN setting, the send (at a transmit power of $0$ dBm) and listen costs for the popular Telos motes~\cite{polastre:telos} are 38mW and 35mW, respectively, and these far exceed the other operations costs for an active device. A similar relationship between the send and listen costs holds for the older, well-known MICA family of devices~\cite{micaz}.

When Carol jams a slot, she disrupts the channel such that no communication is possible; jamming costs $1$ per slot. $\Total$ denotes the total amount Carol will spend over the course of the algorithm; this value is {\it unknown} to either Alice or Bob {\it a priori}.  

In practice, disrupting communication within a slot may be less costly than sending a full message. However, so long as the relative costs for sending, receiving, and jamming are correct to within some (possibly large) constant factor, our asymptotic results hold. Further discussion of this issue, and other practical concerns, is provided in later in Section~\ref{section:discussion}.

If two or more messages are sent within a single slot (a \defn{message collision}), or the slot is jammed, then the slot is said to be \defn{noisy}. If a slot is noisy, this is detectable by a party who is {\it listening} at the {\it receiving end} of the channel, but not by the originator of the transmission. For example, a transmission (jammed or otherwise) from Alice to Bob is detectable only by Bob;  likewise, a  transmission (jammed or otherwise) from Bob to Alice is detectable only by Alice. A party does not know {\it why} a slot is noisy; it may be due to a message collision or to jamming, but the party only learns that the channel is in use. Finally,  a slot which is not noisy and does not contain a single message is said to be {\it clear}. \smallskip 

\noindent{\bf The Communicating Parties:}  If Alice is faulty, there is clearly no hope of communicating $m$; therefore,  Alice is assumed to be correct. In other words, Alice can never be spoofed by the adversary Carol. 

Regarding  Bob,  we define two cases. In Case 1, communications from Bob are always trustworthy. That is, Bob is never spoofed by Carol and we treat Carol as a separate third party. This is a benign case, and it corresponds to situations where communications sent by Bob can be trusted and jamming of $m$ is the only obstacle.


In Case 2, Carol may spoof Bob.\footnote{Equivalently, we can consider that Carol {\it is} Bob, or controls Bob completely. Conceptually, the only two parties are then Alice and Carol}.  We emphasize that this can lead Alice to be uncertain about whether to trust Bob. This  uncertainty corresponds to scenarios where a trusted dealer attempts to disseminate content to its neighbors, some of whom may be spoofed or have suffered a Byzantine fault and are used in an attempt attempt to consume resources by requesting numerous retransmissions. \smallskip

\noindent{\bf The Adversary:} Carol has full knowledge of past actions by Alice and Bob. This allows
for {\it adaptive attacks} whereby Carol may alter her behaviour based on observations she has collected over time. Furthermore, under conditions discussed in Section~\ref{section:reactive_adv}, Carol can also be {\it reactive}: in any slot, she may detect activity on the channel and then jam; however, we assume that she cannot detect when a party is listening.


\subsection{Design Goals for our Algorithm}\label{section:fairfavourable}

In designing our resource-competitive algorithm, we have several goals. In terms of correctness, to reiterate, we want to guarantee (with probability $1$) that Bob receives $m$. 

It is also important that both Alice and Bob have termination conditions. For example, when Bob receives $m$, he should soon terminate such that he can either perform other network tasks, or power down to conserve energy. Similarly, when Alice is certain that Bob has received $m$, she should terminate; from an energy-aware perspective, it is no good to have Alice be unsure of Bob's state and, as a consequence, keep resending $m$ in perpetuity.  

We also desire the following three performance properties with regards to cost.  First, we would like $\rho(T)=o(T)$  such that  Alice and Bob both incur asymptotically less expected cost than Carol when $T$ is large. Jamming attacks are effective because a correct device is often forced to incur a higher cost relative to an attacker. However, if the correct parties incur asymptotically less cost than Carol, then Alice and Bob enjoy the advantage, and Carol is faced with the problem of having her energy resources consumed disproportionately by her attempt to prevent communication. 

Second, we want $\tau$ to be small. This property guarantees that the cost of running our resource-competitive algorithm is low when there is little to no attack.    

Third, we want our results to be \defn{fair}:  Alice and Bob should incur the same worst case asymptotic cost  relative to Carol.  


\subsection{Our Results}\label{section:contributions}

Let $\varphi=(1+\sqrt{5})/2$ denote  the golden ratio. We also draw attention to the well-known relationship that $\Phi = \varphi-1 = 1/\varphi = 0.618...$ where $\Phi$ is known as the {\it golden ratio conjugate}.   Our  main  result is stated below. \vspace{\nval}

\begin{theorem}\label{thm:main1}
Let Carol be an adaptive adversary that jams for $\Total$ slots. There exists a resource-competitive algorithm that guarantees Bob receives $m$, both parties terminate, and has the following properties:\vspace{\nval}
\begin{itemize}
\item{}In Case 1,  the expected cost to Alice and Bob is $O(\Total^{0.5} + 1)$. In Case 2, the expected cost to Alice and Bob is $O(\Total^{\Phi} + 1)=O(\Total^{0.62} + 1)$. \vspace{3pt}

\item{}Both parties terminate in an expected $O(\Total^2)$ and $O(\Total^{\varphi})$ slots for Cases 1 and 2, respectively.
\end{itemize}
\end{theorem}

\noindent{}In other words, we have $\rho(T) = O(T^{0.5})$  and $\rho(T) = O(T^{0.62})$ for Case 1 and Case 2, respectively, and $\tau = O(1)$ for both cases. Later, in Section~\ref{section:reactive_adv}, we demonstrate that Theorem~\ref{thm:main1} still holds when Carol is also reactive so long as there is sufficient network traffic in addition to the sending of $m$.

When Alice wants to send $m$ to $n$ receivers, a similar result is achievable. We consider the  analogue to Case 2 where Carol can spoof any subset of the receivers and/or jam the communications of any/all parties. In this setting, we have the following theorem:

\begin{theorem}\label{thm:main2}
Let Carol be an adaptive adversary that jams for $\Total$ slots. There exists a resource-competitive algorithm that guarantees all receivers obtain $m$, and all parties terminate, and has the following properties that hold with high probability in $n$:\footnote{With probability at least $1-1/n$. We will sometimes use \defn{w.h.p.} as an abbreviation.}
\begin{itemize}
\item{}The cost to Alice and Bob is $O(\Total^{\Phi} + \ln^{\varphi}n)$ and $O(\Total^{\Phi} + \ln n)$, respectively. 

\item{}All parties terminate within $O(\Total^\varphi + \ln^{\varphi}n)$ slots.
\end{itemize}
\end{theorem}

\noindent{\bf Roadmap.} The remainder of this paper is organized as follows. In Section~\ref{section:resource-competitive-oblivious}, we proceed through the design process of our algorithm in order to motivate each step. In Section~\ref{section:analysis}, we analyze our algorithm for Cases 1 and 2, and  prove Theorem~\ref{thm:main1}.   In Section~\ref{section:reactive_adv}, we define the conditions under which we can tolerate a reactive adversary. We then demonstrate how our analysis can be amended to address a reactive adversary.  In Section~\ref{section:multiple}, we generalize the Alice and Bob setting to a general broadcast problem where Alice needs to transmit $m$ to multiple receivers. In Section~\ref{section:discussion}, we provide some additional motivation for our network model. Finally, we conclude with some open problems in Section~\ref{section:conc}.




\section{Algorithm Design and Analysis}\label{section:resource-competitive-oblivious}

We incrementally build towards our resource-competitive algorithm. At each step, our design decisions are explained.\vspace{-5pt}

\subsection{A Naive Attempt}

As a first attempt, Alice and Bob can try to outspend the adversary. For example, let transmission of $m$ be attempted over $\ell$ slots. In each even-indexed slot, Alice sends $m$ while Bob listens. In each odd-indexed slot, if Bob has not received $m$, he sends a negative acknowledgement (\texttt{nack}) message; otherwise, Bob terminates. Alice listens in each odd-indexed slot and, if Alice receives a \texttt{nack}, she  continues onto the next even-indexed slot and sends $m$. Similarly, if Alice detects a noisy odd-indexed slot, she interprets this as the situation where Bob sent \texttt{nack} but the slot was jammed; therefore, she continues with the protocol.  However, if Alice detects a clear odd-indexed slot, she knows Bob received $m$ and terminated since the adversary cannot forge a clear slot; in this case, Alice  safely terminates.

While Alice and Bob both finish correctly, note that if the adversary jams $\Total$~consecutive even-indexed slots, then Alice and Bob each send and listen for $2T+2$ slots. Therefore, Alice and Bob {\it each} spend more than twice what the adversary spends; that is, $\rho(T)>2T$ and the adversary rapidly disables each party by depleting the respective energy supplies. This illustrates  why jamming is often an effective attack.

\begin{figure}[t]
\begin{center}
\fbox{\parbox[t]{3.4in}{\noindent{}\LN~for any round\vspace{6pt}

\noindent
\emph{Send Phase:} For each of $\ell$ slots do\vspace{-2pt}
\begin{itemize}
\item Alice sends $m$ with probability $2/\sqrt{\ell}$\vspace{-0pt}
\item Bob listens with probability $2/\sqrt{\ell}$ \vspace{-0pt} 
\end{itemize}

\noindent
If Bob received $m$, then he terminates\\\vspace{-3pt} 

\noindent
\emph{Nack Phase:} For $1$ slot do\vspace{-2pt}
\begin{itemize}
\item Bob sends a \texttt{nack} message\vspace{-0pt}
\item Alice listens 
\end{itemize}

\noindent
If Alice listened to a clear slot, then she terminates

} 

} 
\end{center}
\vspace{-14pt}
\caption{Pseudocode for \LN.}\label{fig:pseudocode_it1}
\vspace{-10pt}\end{figure}

\subsection{Towards a Resource-Competitive Guarantee}\label{sec:less-naive}

An initial attempt at a resource-competitive approach is \LN~in Figure~\ref{fig:pseudocode_it1}. In the Send Phase, Alice and Bob  send and listen each with probability $2/\sqrt{\ell}$. If Bob ever receives $m$, he terminates the protocol. In the Nack Phase, if Bob has not terminated, he sends \texttt{nack} to Alice during the single slot  asking her to enter into a new round. If Alice hears \texttt{nack}, or if the slot is jammed, she proceeds into the next round; otherwise, if the slot is clear, she terminates. Let a \defn{round} refer to the execution of a Send Phase and the corresponding Nack Phase.

Using a birthday-paradox-like argument, there is likely to be a non-jammed slot where both Alice sends $m$ and Bob listens; this is true even if a constant fraction, say $1/2$, of the slots are jammed.  Therefore, communication  likely  succeeds unless the adversary jams more than half of the slots. In a Send Phase where more than half of the slots are jammed -- referred to as a \defn{blocked Send Phase} --- then Bob may not receive $m$. But now $T=\Omega(\ell)$ while both Alice and Bob spend only $O(\sqrt{\ell}) = O(\sqrt{T})$ in expectation. Therefore, to prevent communication in the round, the adversary must incur a cost that is roughly quadratically larger.  This is exactly the flavor of result that we seek.

Are we done? No, there are two shortcomings to \LN~and we describe the first here. Let us focus on the Send Phase and consider what happens if the adversary forces $k$ consecutive rounds with blocked Send Phases. The adversary spends $T=\Omega(k\ell)$ while each party spends $O(k\sqrt{\ell}) = O(T/\sqrt{\ell})$ in expectation. Therefore, for $\Total > \ell$, we no longer have the quadratic advantage obtained above for a single blocked Send Phase. 

This problem arises because each blocked Send Phase imposes equal asymptotic cost on the adversary. In contrast, imagine if the cost of the final blocked Send Phase equaled the total cost  of all previous $O(k)$ blocked Send Phases. Then, the quadratic advantage would still hold by the same reasoning as above; simply apply the same analysis to the final blocked phase. We can achieve this property by increasing the length of the Send Phase as an exponential function of the number of rounds completed so far.

Figure~\ref{fig:less-naive-mod} provides the pseudocode with this modification implemented; we name this intermediate algorithm \MLN. The length of the Send Phase is $2^{ci}$ where the constant $c>0$ is left unfixed (to be as general as possible) and will be determined later.

\begin{figure}[t]
\begin{center}
\fbox{\parbox[t]{3.4in}{\noindent{}\MLN~for round $i\geq 1$\vspace{5pt}

\noindent
\emph{Send Phase:} For each of the $2^{ci}$ slots do\vspace{-0pt}
\begin{itemize}
\item Alice sends $m$ with probability $2/2^{(c-a)i}$\vspace{-0pt}
\item Bob listens with probability $2/2^{(c-b)i}$\vspace{-0pt}
\end{itemize}
If Bob received $m$, then he terminates\\\vspace{-4pt}

\noindent
\emph{Nack Phase:} For $1$ slot do\vspace{-2pt}
\begin{itemize}
\item Bob sends a \texttt{nack} message\vspace{-0pt}
\item Alice listens
\end{itemize}

\noindent
If Alice listened to a clear slot, then she terminates.

} 

} 
\end{center}
\vspace{-14pt}
\caption{Pseudocode for modified~\MLN.}\label{fig:less-naive-mod}
\vspace{-10pt}\end{figure}

The probabilities of Alice sending and Bob listening are also modified such that the expected cost to Alice is $O(2^{ai})$ and to Bob is $O(2^{bi})$, where $a>0$ and $b>0$ are constants that we leave unspecified for now.  However, given that a birthday-paradox-like argument is used to show that Alice and Bob succeed in communication, then we require $a+b = c$ to make our analysis work.\footnote{More generally, we need $a+b\geq c$, but we seek to minimize cost and so we opt for equality.}  

To illustrate what is meant by such an argument,  we \underline{\it sketch} some analysis for $c=2$ and $a=b=1$. For simplicity,  assume a weaker adversary who jams each slot with probability $p_j$. Later, in Section~\ref{section:analysis}, we provide a rigorous analysis that holds against our more powerful adversary.  

First, consider the case of ``light  jamming'', say $p_j \leq 1/2$. Since $a=b=1$, Alice sends and Bob listens with probability $\Theta(1/2^{i})$ per slot. The probability that Bob receives $m$ in slot $j$ is:
\begin{eqnarray*}
&=& Pr(\mbox{slot $j$ not jammed}) \times Pr(\mbox{Alice sends in slot $j$})\\
 && \times Pr(\mbox{Bob listens in slot $j$}) \\
 & = & (1-p_j)  \left( \frac{2}{2^{i}} \right)\left( \frac{2}{2^{i}} \right)
\end{eqnarray*}

The probability that Bob fails to receive $m$ over all $2^{ci} = 2^{2i}$ slots in the Send Phase is at most:
 $$\left(1 - \frac{4(1-p_j) }{2^{2i}} \right)^{2^{2i}} \leq e^{-2}$$ 
\noindent  Therefore, when (roughly) at most half the slots are jammed, there is a constant probability of success, but this argument requires $a+b = c$.

Conversely, consider $p_j >1/2$; for example, perhaps $p_j=1$ and all slots are jammed in the Send Phase. Clearly, Alice will fail in sending $m$ to Bob. However, we do ``succeed'' in making Carol incur significantly more cost. The expected cost to Carol is $T =\Omega(2^{ci})$ while Alice and Bob each spend $O(2^i) = O(\sqrt{T})$ in expectation. Therefore, we recover the quadratic result above for a blocked Send Phase. 

Finally, we highlight the fact that $c>2$ does not improve this advantage. For example, if we set $c=4$, then we require $a+b=4$ and the maximum expected cost to either party is minimized by setting $a=b=2$ and this still yields an expected cost to Alice and Bob of  $O(\sqrt{T})$.

\begin{figure}[t]
\begin{center}
\fbox{\parbox[t]{3.4in}{\noindent{}\AT~for round $i\geq 1$\vspace{5pt}

\noindent
\emph{Send Phase:} For each of the $2^{ci}$ slots do\vspace{-0pt}
\begin{itemize}
\item Alice sends $m$ with probability $2/2^{(c-a)i}$\vspace{-0pt}
\item Bob listens with probability $2/2^{(c-b)i}$\vspace{-0pt}
\end{itemize}
If Bob received $m$, then he terminates\\\vspace{-4pt}

\noindent
\emph{Nack Phase:} For each of the $\ell$ slots do\vspace{-2pt}
\begin{itemize}
\item Bob sends a \texttt{nack} message\vspace{-0pt}
\item Alice listens with probability $4/\ell$\vspace{-0pt}
\end{itemize}

\noindent
If Alice listened to a clear slot, then she terminates.

} 

} 
\end{center}
\vspace{-14pt}
\caption{Pseudocode for \AT.}\label{fig:almosthere}
\vspace{-10pt}\end{figure}


\subsection{Robustness to Attacks in the Nack Phase}\label{sec:almost}

The second shortcoming of \LN~(and also of \MLN) involves the Nack Phase under Case 2. Recall that Bob can be spoofed and his communications can be jammed. Assume that Bob receives $m$ and terminates. Then, in the very next Nack Phase, the adversary may spoof a \texttt{nack} message and force Alice to execute another round. In each subsequent Send Phase, the adversary will do nothing (in order to avoid any cost) and then send another \texttt{nack} in the next Nack Phase.  Therefore, in each round $i$, Alice incurs a cost of roughly $2^{ai}$ in expectation while the adversary has a cost of $1$. Through this attack, the adversary can force Alice to quickly deplete her energy supply. 

Note that even if messages from Bob could be authenticated, the adversary may simply generate noise which will yield the same result.  This is unavoidable since collision detection is used as a reliable negative acknowledgement.

This problem arises because it is ``too cheap'' for Bob to force Alice to proceed into the next round. Intuitively, a remedy is to increase the cost required for Bob to prevent Alice's termination. As with the Send Phase of \LN, this could be accomplished by setting the length of the Nack Phase to be, say, $\ell$ and requiring Bob to send a \texttt{nack} in each slot. Therefore, Bob is making a ``down payment" and this should be large enough to deter the adversary from spoofing Bob in the Nack Phase. 

Should Alice verify the down payment by listening to all $\ell\geq 4$ slots? She is already incurring an expected cost of $2^{ai}$ from the preceding Send Phase, so she could afford to spend the same in the Nack Phase without changing her asymptotic cost. But what happens if the algorithm runs sufficiently long such that $2^{ai}$ exceeds $\ell$? Instead, as we show in Section~\ref{section:analysis}, Alice can sample each slot with probability $4/\ell$ in the Nack Phase; this only contributes additional expected cost of $4$.  The adversary may attempt to spend less than $\ell$ whilst spoofing Bob, but Alice's random strategy is adequate to thwart such behavior. The pseudocode for this modification is given in Figure~\ref{fig:almosthere}.

Finally, we note that fixing the length of the Nack Phase to be $\ell$ in each round runs into a problem similar to that faced in the Send Phase of \LN: multiple consecutive blocked Nack Phases -- where more than half the slots are jammed -- will degrade the resource-competitive ratio. Therefore, we adopt the same solution by increasing the length of the Nack Phase as a function of the number of rounds completed so far. Since Bob spends an expected $2^{bi}$ in the Send Phase, the length of the Nack Phase can also be $2^{bi}$ without increasing his asymptotic cost.

The pseudocode for these modifications is presented as \MAT~in Figure~\ref{fig:pseudocode_2party}. 

\begin{figure}[t]
\begin{center}
\fbox{\parbox[t]{3.4in}{\noindent{}\MAT~for round $i\geq 2/b$\vspace{5pt}

\noindent
\emph{Send Phase:} For each of the $2^{ci}$ slots do\vspace{-0pt}
\begin{itemize}
\item Alice sends $m$ with probability $2/2^{(c-a)i}$\vspace{-0pt}
\item Bob listens with probability $2/2^{(c-b)i}$\vspace{-0pt}
\end{itemize}
If Bob received $m$, then he terminates\\\vspace{-4pt}

\noindent
\emph{Nack Phase:} For each of the $2^{bi}$ slots do\vspace{-2pt}
\begin{itemize}
\item Bob sends a \texttt{nack} message\vspace{-0pt}
\item Alice listens with probability $4/2^{bi}$\vspace{-0pt}
\end{itemize}

\noindent
If Alice listened to a clear slot, then she terminates.

} 

} 
\end{center}
\vspace{-14pt}
\caption{Pseudocode for \MAT.}\label{fig:pseudocode_2party}
\vspace{-10pt}\end{figure}


\subsection{Handling Variance of Cost}\label{sec:variance}

Consider the following strategy by the adversary. All slots in each consecutive Send Phase are jammed until a round $r$ is encountered where Alice sends $2^{cr}$ in the corresponding Send Phase. Note that Alice sends probabilistically, so there is a small, but non-zero, probability of this event occurring in each round, and the adversary's strategy means it will eventually occur with probability $1$. At that point, $T=\Theta(2^{cr})$ and Alice's cost is $\Theta(T)$ which is clearly poor.

Why not modify the algorithm such that Alice randomly selects exactly $2^{ai+1}$ slots in which to send prior to executing the next Send Phase? This prevents any variance in cost and, therefore, prevents the above attack.

To see why this is a bad idea, note an advantage of having Alice decide randomly whether to send on a per-slot basis: {\it an adaptive adversary has no more power than a non-adaptive one}. This is because the action Alice takes in the current slot is independent of what she has done in the previous slots. In other words, knowledge of the players' past actions does not help Carol plan her jamming in the next slot. Therefore, this design choice provides a useful property for tolerating an adaptive adversary.

Selecting a fixed number of slots to send in {\it a priori} deprives us of this property.  For example, if Alice does not send $m$ in slot 1 of the Send Phase, an adaptive adversary learns of this prior to its decision on whether to jam in slot 2. The absence of sending in slot 1 implies that the probability that Alice sends in slot 2 has increased given that she must send in exactly $2^{ai+1}$ slots; that is, these events are dependent.  For the same reason, the probability of a slot in which Alice sends {\it and} Bob listens changes as we proceed through the Send Phase. The adversary may use this information to improve its jamming and, at the very least, analyzing such an algorithm seems difficult. 

Instead, we make the following modification. In the Send Phase, if Alice has sent in $2^{ai+2}$ slots, she remains silent for the remainder of the phase. Similarly, if Bob has listened in $2^{bi+2}$ slots, he performs no more listening for the remainder of the phase. This \defn{cutoff} for either party bounds the cost
to be at most a constant factor more than the expectation. In other words, the cost to Alice and Bob is $O(2^{ai})$ and $O(2^{bi})$, respectively (the additional constant factors are useful in applying a Chernoff bound later on in the analysis). The same modification is made to the Nack Phase: if Alice has listened in $2^{ai+2}$ slots, she performs no more listening for the remainder of the phase. The pseudocode for this final design decision is presented in Figure~\ref{fig:final-alg} as \Alg.

This design maintains independence between events in each slot up until the cutoff point. Also, in each slot up until this cutoff point, the probability of a slot where Alice sends and Bob listens is always the same at $4/2^{(c-a)i+(c-b)i}$. Of course, once the cutoff point is reached, the adversary learns this and does not need to perform any more jamming; however, up to this point, we preserve the robustness to an adaptive adversary.  In analyzing the impact of either party reaching the cutoff in a round, we will pessimistically assume failure and incorporate the cost of this round for each party. However, in Section~\ref{section:analysis}, we will show that this a low-probability event and does not impact the expected cost.

\begin{figure}[t]
\begin{center}
\fbox{\parbox[t]{3.4in}{\noindent{}\Alg~for round $i\geq 2/b$\vspace{5pt}

\noindent
\emph{Send Phase:} For each of the $2^{ci}$ slots do\vspace{-0pt}
\begin{itemize}
\item If Alice has sent in less than $2^{ai+2}$ slots in this phase, she sends $m$ with probability $2/2^{(c-a)i}$\vspace{2pt}
\item If Bob has listened in less than $2^{bi+2}$ slots in this phase, he listens with probability $2/2^{(c-b)i}$\vspace{-0pt}
\end{itemize}
If Bob received $m$, then he terminates\\\vspace{-4pt}

\noindent
\emph{Nack Phase:} For each of the $2^{bi}$ slots do\vspace{-2pt}
\begin{itemize}
\item Bob sends a \texttt{nack} message\vspace{2pt}
\item If Alice has listened in less than $2^{ai+2}$ slots in this phase, she listens with probability $4/2^{bi}$\vspace{2pt}
\end{itemize}

\noindent
If Alice listened to a clear slot, then she terminates.

} 

} 
\end{center}
\vspace{-14pt}
\caption{Pseudocode for \Alg.}\label{fig:final-alg}
\vspace{-10pt}\end{figure}


 \subsection{Description of~\Alg} \label{sec:description}

We summarize a round $i$ of \Alg:\smallskip\smallskip

\noindent{$\bullet$~}{\it Send Phase:} This phase consists of $2^{ci}$ slots. If Alice has sent in less than $2^{ai+2}$ slots, then she will send $m$ in the current slot with probability $\frac{2}{2^{(c-a)i}}$. This yields an expected cost of $2^{(c-a)i+1}$ over the entire phase. If Bob has listened in less than $2^{bi+2}$ slots, he will listen in the current slot with probability $\frac{2}{2^{(c-b)i}}$. This yields an expected total cost of $2^{bi+1}$ over the entire phase.\vspace{5pt}

\noindent{$\bullet$~}{\it Nack Phase:} This phase consists of $2^{bi}$ slots. If Bob has not received $m$, then he sends a \texttt{nack} in all $2^{bi}$ slots.  If Alice has listened in less than $2^{ai+2}$ slots, then she will  listen in the current with probability $4/2^{bi}$ (note that $i\geq{}2/b$ is required) for an expected total cost of $4$ over the phase.\vspace{5pt}

\noindent{\bf Termination Conditions:} As discussed in Section~\ref{section:fairfavourable}, termination conditions are important and we highlight these here.  Bob terminates the protocol upon receiving $m$. Since Alice cannot be spoofed, as discussed in Section~\ref{section:attack_model}, this termination condition suffices. 

Alice terminates if she listens to a clear slot  (neither noisy nor containing a \texttt{nack} message) in the Nack Phase; since jammed slots are detectable by Alice while listening (Section~\ref{section:attack_model}), this condition suffices. That is, Alice continues into the next round if and only if (i) Alice listens to zero slots or (ii) all slots listened to by Alice in the Nack Phase contain a blocked slot or \texttt{nack} message. There are two situations where this occurs:\vspace{2pt}

\noindent{$\bullet$~}{\it Send Failure:} Bob has not received $m$.\vspace{2pt}

\noindent{$\bullet$~}{\it Nack Failure:}   Bob has terminated  and Carol either spoofs \texttt{nack} messages or jams slots in order to trick Alice into thinking a valid \texttt{nack} was sent but jammed.\smallskip\smallskip

\noindent{\bf Nack Failures and Cases 1 \& 2:} Note that an ``acknowledgement'' is indicated by having at least one clear slot in the Nack Phase.  A \defn{Nack Failure} refers to the situation in which, by the end of the Nack Phase, Alice believes Bob is alive and has failed to receive $m$. 

In Case 2, as discussed in Section~\ref{sec:almost}, a Nack Failure may occur via an attack in the Nack Phase after Bob has received $m$ and terminated. Carol may spoof Bob in the Nack Phase by making it appear as if Bob did not receive $m$ and is requesting that both he and Alice proceed into the next round. This behavior keeps Alice active and incurring a cost. Critically, this attack affects Alice only; if Bob has not terminated, a \texttt{nack} message will be issued anyway and the attack accomplishes nothing. Therefore, we need only consider this attack in the situation where Bob has terminated.

In Case 1, no jamming occurs in the Nack Phase and, therefore, no Nack Failure can occur. We note that, in Case 1, the Nack Phase can be shortened to a single slot -- the algorithm~\MLN~will suffice -- where Bob sends his \texttt{nack} message and Alice listens; however,  this does not change our {\it asymptotic} cost for Case 1. Since our current presentation applies to both cases, we proceed with analyzing \Alg.\vspace{\bigval}


\section{Analysis of~\Alg}\label{section:analysis}\vspace{\nval}

Throughout, assume ceilings on the number of slots indicated for use by Alice or Bob if it is not an integer. For any given round, we say it is a \defn{blocked Send Phase} if Carol jams at least half of the slots in the Send Phase; otherwise, it is a  \defn{non-blocked Send Phase}. Similarly, a \defn{blocked Nack Phase} occurs if  Carol jams or spoofs \texttt{nack} messages from Bob in at least half the slots in the Nack Phase; otherwise, it is a \defn{non-blocked Nack Phase}. \smallskip\smallskip

\noindent{\bf Bounds on constants \boldmath{$a,b,c$}:} We make a few important remarks about these constants. Note that if $a\geq{}b$, then the expected cost to Alice is at least as much as the expected cost to Bob. But recall that, in Case 2, the adversary can spoof Bob. Therefore, this allows Carol to cause a Nack Failure with at most the same cost as what Alice incurs in the preceding Send Phase (note that Carol will avoid listening in the Send Phase and incur zero cost there). As discussed in Section~\ref{sec:almost}, we must avoid this since it admits an energy-draining attack against Alice. Therefore, in Case 2, we require $a < b$.

Additionally, our analysis will depend on a birthday-paradox-like argument (in Lemmas~\ref{lemma:nonjamming1} and \ref{lemma:nonjamming2}) to show that $m$ or a \texttt{nack} message is successfully transmitted across the channel. As sketched in Section~\ref{sec:less-naive}, we will need that $a+b = c$ for the analysis to work. 

Finally, as discussed in Section~\ref{sec:less-naive}, we assume that $c \leq 2$.  Consequently, $a < 2$ and $b < 2$.\smallskip\smallskip

\noindent{}Our analysis is composed of two pieces. First, we consider ``light jamming'' where we have no blocked phases. We show that Alice succeeds in communicating $m$ to Bob with probability at least $1-e^{-2}$ in each non-blocked Send Phase. Given that Bob has terminated, a similar result is shown for the probability that Alice (correctly) terminates in a non-blocked Nack Phase.

Second, we consider ``heavy jamming'' where we address  blocked phases. Given the exponentially increasing length of the phases, intuitively the expected cost is dominated by the cost to the parties in the last blocked round.  We can then use this cost to compare against Carol's cost $\Total$.

In both situations, the issue of variance described in Section~\ref{sec:variance} adds a subtlety to our calculations. For light jamming, we only show constant probability of success in non-blocked phases so long as neither Alice nor Bob reach their respective cutoffs.

For heavy jamming, note that the adversary decides when the last blocked round occurs. For example, what if the adversary jams all slots until the (low-probability) event that Alice reaches her cutoff point in the Send Phase?  This highlights the problem of conditioning on the adversary's final blocked round in order to calculate the expected cost. 

Our modifications in Section~\ref{sec:variance} permit a clean analysis. We bound the number of rounds in which either party reaches its respective cutoff. We can show that this occurs with probability exponentially small in the round index $i$, and therefore the impact on the expected cost to each party is negligible. We make use of the following Chernoff bound:

\begin{theorem}\label{thm:Chernoff}(\cite{motwani_raghavan:randomized})
Let $X_1,\dots, X_n$ be binary random variables such that $\Pr(X_i) = p$ and let $X = \sum_{i=1}^{n} X_i$. 
For any $\delta$, where $0 < \delta < 1$\mbox{$:$}
$$\Pr(X>(1+\delta)\,E[X]) \leq e^{-\delta^2 \,E[X] / 3}$$
\end{theorem}

Define a \defn{failed round} to be a round in which either Alice or Bob reaches a cutoff in either the Send or Nack Phase; otherwise, it is \defn{non-failed round}. We now  bound the probability of a failed round as a function of the round index $i$:

\begin{lemma}\label{lemma:deviation}
For any round $i\geq 1$ the probability of a failed round is at most $\exp(-2^{ai}/3) + \exp(-2^{bi}/3)$.
\end{lemma}
\begin{proof}
In the Send Phase of round $i$, let the binary random variable $X_j=1$  if Alice sends $m$ in slot $j$; otherwise, $X_j=0$. Letting $X=\sum_{j=1}^{2^{ci}} X_j$, then:
$$E[X] = \sum_{j=1}^{2^{ci}} E[X_j] = \sum_{j=1}^{2^{ci}} (2/2^{(c-a)i}) = 2^{ai+1}$$

\noindent by linearity of expectation.  A similar calculation for the Nack Phase yields an expected cost of $4$. Therefore, over epoch $i$, Alice's expected cost is $2^{ai+1} + 4$ for $i\geq 1$.  By Theorem~\ref{thm:Chernoff}, using $\delta = 1$ yields that the probability of exceeding this expected cost by a factor of at least $2$ is less than $e^{-2^{ai}/3 }$.

In the Send Phase of round $i$,  let the binary random variable $Y_j=$  if Bob listens in slot $j$; otherwise, $Y_j=0$. Letting $Y=\sum_{j=1}^{2^{ci}} Y_j$, then:
$$E[Y] = \sum_{j=1}^{2^{ci}} E[Y_j] = \sum_{j=1}^{2^{ci}} (2/2^{(c-b)i}) = 2^{bi+1}$$ 

\noindent by linearity of expectation. There is an additional $2^{bi}$ cost associated with the Nack Phase. By Theorem~\ref{thm:Chernoff},  the probability of exceeding the expected cost by a factor of at least 2 is less than $e^{-2^{bi}/3 }$.\qed
\end{proof}

\noindent By Lemma~\ref{lemma:deviation}, the probability of a failed round is very small as a function of the round index $i$. We will show that the contribution of failed rounds does not impact our asymptotic analysis.

We now present the light jamming portion of our analysis.

\begin{lemma}\label{lemma:nonjamming1} Consider a blocked Send Phase in a non-failed round. The probability that Bob does not receive the message from Alice  is at most $e^{-2}$.
\end{lemma}\vspace{\pval}
\begin{proof}
Let $s=2^{ci}$ be the number of slots in the Send Phase. Let $p_A$ be the probability that Alice sends in a particular slot. Let $p_B$ be the probability that Bob listens in a particular slot. Let $X_j=1 $ if the message is not delivered from Alice to  Bob in the $j^{th}$ slot. Then: 
\begin{eqnarray*}
&& Pr[\mbox{$m$ is not delivered in the Send Phase}]\\
&=&Pr[X_1 X_2\cdots{}X_s=1]\\
&=& Pr[X_s=1~|~ X_1~X_2\cdots{}X_{s-1}=1]\cdot{}\prod_{i=1}^{s-1}Pr[X_i=1]
 \end{eqnarray*}
  Let $q_j=1$ if Carol does not jam in slot $j$; otherwise, let $q_j=0$.  The value of $q_j$ can be selected arbitrarily by Carol. Then:
\begin{eqnarray*}  
Pr[X_i=1~|~ X_1 X_2\cdots{}X_{i-1}=1] &=& 1-p_Ap_Bq_j 
\end{eqnarray*}
\noindent Substituting for each conditional probability, we have:
\begin{eqnarray*}  
Pr[X_1X_2\cdots{}X_s=1] &=&(1-p_Ap_Bq_1)\cdots{}(1-p_Ap_Bq_s)\\
&=& \prod_{j=1}^{s}(1-p_Ap_Bq_j)\\
&\leq& e^{-p_Ap_B \sum_{j=1}^s q_j}\\
&\leq& e^{-2} 
\end{eqnarray*}
\noindent where the last inequality follows from:
\begin{eqnarray*}  
p_Ap_B \sum_{j=1}^s q_j  & \geq &  (2/2^{(c-a)i})(2/2^{(c-b)i}) (s/2) \\
& = & (2/2^{(c-a)i})(2/2^{(c-b)i})(2^{ci}/2)\\
&=& 2 
\end{eqnarray*} 
\noindent since $a+b = c$ and there is no blocked send phase which means Carol jams at most $s/2$ slots.\qed
\end{proof}

Note that Lemma~\ref{lemma:nonjamming1} handles adaptive (but not reactive) adversaries. A simple but critical feature of tolerating adaptive adversaries is that the probability that a party is active in one slot is independent from the probability that the party is active in another slot. Therefore, knowing that a party was active for $k$ slots in the past conveys no information about future activity.  For reactive adversaries, we need only modify Lemma~\ref{lemma:nonjamming1} and we address this later (see Section~\ref{section:reactive_adv}).

\begin{lemma}\label{lemma:case1}
Assume there are no blocked Send Phases and no blocked Nack Phases.  The expected cost of each party is $O(1)$.
\end{lemma}\vspace{\pval}
\begin{proof}
We compute the expected cost over both non-failed and failed rounds for $i\geq 1$; note that this can only overestimate the expected cost since \Alg~specifies $i\geq 2/b$ and $b<2$. 

Using Lemmas~\ref{lemma:deviation} and~\ref{lemma:nonjamming1}, the expected cost to Alice is at most:\vspace{-15pt}

\begin{eqnarray*}
\hspace{1cm}&& \sum_{i=1}^{\infty}  e^{-2(i-1)}\,O(2^{ai}) +  \sum_{i=1}^{\infty} \exp{\left(-2^{\Theta(i-1)}\right)} O(2^{ai}) \\  
& = & O(1)\sum_{i=1}^{\infty} \left(\frac{2^a}{e^2}\right)^i \\ 
&=& O(1)  
\end{eqnarray*}
 
\noindent where the last line follows from the fact that $a < 2$ and the sum of a geometric series. Similarly, the expected cost to Bob is at most:
\begin{eqnarray*}
\hspace{1cm}&& \sum_{i=1}^{\infty}  e^{-2(i-1)} O(2^{bi})  +  \sum_{i=1}^{\infty} \exp{\left(-2^{\Theta(i-1)}\right)} O(2^{bi}) \\
&\leq& O(1) \sum_{i=1}^{\infty} \left(\frac{2^b}{e^2}\right)^i \\
&=&O(1)
 \end{eqnarray*} 
 
\noindent where the last line follows from the fact that $b < 2$ and the sum of a geometric series.~$\qed$
\end{proof}

\begin{lemma}\label{lemma:nonjamming2} Assume that Bob has received $m$ by non-failed round $i$, and that round $i$
has a non-blocked Nack Phase. Then, the probability of a Nack Failure is at most $e^{-2}$.
\end{lemma}\vspace{\pval}
\begin{proof}
Let $s=2^{bi}$ denote the number of slots in the Nack Phase and let $p=4/2^{bi}$ denote the probability that Alice listens in a slot.
For slot $j$, define $X_j$ such that $X_{j}=1$ if Alice does not terminate.  Then, $Pr[$ Alice retransmits $m$ in round $i+1] = $ $Pr[X_1X_2\cdots{}X_s=1]$. Let $q_j=1$ if Carol does not jam  in slot $j$; otherwise, let $q_j=0$. The $q_j$ values are determined arbitrarily by Carol. Since Alice terminates when she listens and hears a clear slot, then $Pr[X_j=1]=(1-pq_j)$. Therefore, $Pr[X_1X_2\cdots{}X_s=1] \leq e^{-p\sum_{j=1}^{s}q_j} \leq e^{-2}$.\qed
\end{proof}

\noindent We now consider the heavy jamming portion of our analysis.

\begin{lemma}\label{lemma:costs}
Assume there is at least one blocked Send Phase.  In Case 1, the expected cost to Alice and Bob is $O(\Total^{a/c})$ and $O(\Total^{b/c})$, respectively. In Case 2, the expected cost to Alice and Bob is $O(\Total^{a/c}+\Total^{a/b})$ and  $O(\Total^{b/c})$, respectively.
\end{lemma}
\begin{proof}
\noindent{}Let $i\geq{}2/b$ be the last blocked Send Phase. Let $j\geq{}i$ be the last blocked Nack Phase; if no such blocked Nack Phase exists, then assume $j=0$. 

Carol may choose $i$ and $j$ depending on the history of the execution. With this notation in place, we can bound the cost to Carol, Alice, and Bob. \smallskip\smallskip

\noindent{\bf Carol:} In Case 1,  only blocked Send Phases occur and so $\Total= \Omega(2^{ci})$. In Case 2,  since there can be a blocked Nack Phase, the total cost to Carol is $\Total = \Omega(2^{ci} + 2^{bj})$. \smallskip\smallskip

\noindent{\bf Alice:} We begin by calculating the expected cost to Alice prior to successfully transmitting $m$.   Using Lemmas~\ref{lemma:deviation}~and~\ref{lemma:nonjamming1}, the expected cost to Alice prior to $m$ being delivered is at most:
\begin{eqnarray*}
&& O(2^{ai}) + \sum_{k=1}^{\infty}  e^{-2(k-1)}  \,O(2^{a(i+k)}) \\
& & \hspace{30pt}+ \sum_{k=1}^{\infty} \exp{\left(-2^{\Theta(k-1)}\right)} O(2^{a(i+k)})\\
& = & O(2^{ai}) +  O(2^{ai})\sum_{k=1}^{\infty} \left(\frac{2^{a}}{e^2}\right)^k +O(2^{ai})   \\
&=&  O(2^{ai}) 
\end{eqnarray*}
\noindent  where the last line follows $a<2$ and the sum of a geometric series.

Next, using Lemmas~\ref{lemma:deviation} and~\ref{lemma:nonjamming2}, we calculate the expected cost to Alice after delivery of $m$; this addresses blocked Nack Phases possible only in Case 2. By assumption, the last blocked Nack Phase occurs in round $j$ and therefore Alice's expected cost is at most:
\begin{eqnarray*}
&& O(2^{aj}) + \sum_{k=1}^{\infty} e^{-2(k-1)} O(2^{a(j+k)})\\
& & \hspace{30pt}+ \sum_{k=1}^{\infty} \exp{\left(-2^{\Theta(k-1)}\right)} O(2^{a(j+k)})\\
& = & O(2^{aj}) +  O(2^{aj})\sum_{k=1}^{\infty} \left(\frac{2^{a}}{e^2}\right)^k + O(2^{aj})   \\
&=&O(2^{aj}) 
\end{eqnarray*}

\noindent where the last line follows $a<2$ and the sum of a geometric series.

Therefore, in Case 2, the total expected cost to Alice is $O(2^{ai} + 2^{aj}) $.  Since $\Total = \Omega(2^{ci}  + 2^{bj} )$, this cost as a function of $\Total$ is $O(\Total^{a/c}+\Total^{a/b})$. For Case 1, there are no blocked Nack Phases and so  Alice's cost is simply $O(\Total^{a/c})$.\smallskip\smallskip

\noindent{\bf Bob:} We analyze the expected cost to Bob up until $m$ is received. Note that we assume Carol does not have any cost for blocked Nack Phases up until this point since, as discussed in Section~\ref{sec:description}, there is no benefit to causing a Nack Failure via jamming. 

Using Lemma~\ref{lemma:nonjamming1},  Bob's expected cost prior to receiving $m$ is at most:
\begin{eqnarray*}
&& O(2^{bi}) + \sum_{k=1}^{\infty}  e^{-2(k-1)}O(2^{b(i+k)}) \\
& & \hspace{30pt}+ \sum_{k=1}^{\infty} \exp{\left(-2^{\Theta(k-1)}\right)} O(2^{b(i+k)})\\
& \leq & O(2^{bi}) +  O(2^{bi})\sum_{k=1}^{\infty} \left(\frac{2^{b}}{e^2}\right)^k   + O(2^{bi})\\
&=& O(2^{bi}) 
\end{eqnarray*}
\noindent where the last line follows $b<2$ and the sum of a geometric series. Therefore, the expected cost for Bob as a function of $\Total = \Omega(2^{ci})$ is $O(\Total^{b/c})$. $\qed$
\end{proof}\vspace{\nval}

\noindent{}We now prove Theorem~\ref{thm:main1} stated in Section~\ref{section:contributions}:\vspace{\myval}

\noindent{\it Proof of Theorem~\ref{thm:main1}:} For both Case 1 and 2, when there are no blocked Send Phases and no blocked Nack Phases, the expected cost to each party is $O(1)$ by Lemma~\ref{lemma:case1}. Therefore, we have $\tau = O(1)$.

For both Case 1 and Case 2, if there is at least one blocked Send Phase, Lemma~\ref{lemma:costs} gives the expected cost for Alice and Bob as $O(\Total^{a/c})$ and $O(\Total^{b/c})$, respectively. Recall from Section~\ref{section:fairfavourable} that we desire the worst-case relative costs for Alice and Bob be equal. Therefore, setting $a/c = b/c$ implies that $a=b=1$ and $c=2$. This yields $\rho(\Total) = O(\Total^{0.5})$ and therefore the expected cost to each party is $O(\Total^{0.5} + 1)$.

In Case 2,  if there is at least one blocked Send Phase or Nack Phase, the expected cost to Alice and Bob is $O(\Total^{a/c}+\Total^{a/b})$ and $O(\Total^{b/c})$, respectively,  by Lemma~\ref{lemma:costs}.  The exponents of interest are now $a/c$, $a/b$, and $b/c$. Notice that $a/c$ and $a/b$ both correspond to Alice's expected cost, but that $a/c < a/b$. Therefore, we need consider only $a/b$ and $b/c$ and we wish $a/b = b/c$ to obtain fairness. Also recall that we insisted on $a+b = c$.  

We now have three parameters and two constraints. However, this is sufficient for us to derive useful values for $a,b,$ and $c$. By dividing the second constraint by the first, we have:
\begin{eqnarray*}
&& \frac{(a+b)b}{a} = \frac{c^2}{b} \\
&\Rightarrow& \left(1 + \frac{b}{a} \right) = \left(\frac{c}{b}\right)^2\\
&\Rightarrow& \left(1 + \frac{c}{b} \right) = \left(\frac{c}{b}\right)^2 \mbox{~~~since~} \frac{a}{b} = \frac{b}{c}\\
&\Rightarrow& \left(\frac{c}{b}\right)^2 - \left(\frac{c}{b} \right)  - 1 = 0
\end{eqnarray*}
\noindent Solving for the roots of this quadratic polynomial yields $c/b = (1+\sqrt{5})/2 = \varphi$ which is the golden ratio. Therefore, so long as the constants $a,b,c$ are set such that $c/b = b/a = \varphi$, \Alg~has $\rho(T) = O(\Total^{\Phi}) = O(\Total^{\varphi-1}) = O(\Total^{0.62})$ where $\Phi = 1/\varphi = \varphi-1$. 

How many slots in expectation occur prior to termination by both Alice and Bob? We focus on Alice since she terminates after Bob. Assume pessimistically that a blocked Send or Nack Phase prevents her from terminating. Carol can stretch her budget the farthest by blocking the Nack Phase only as this incurs a cost of $2^{bi}/2$ rather than a cost of $2^{ci}/2$ for blocking the Send Phase. 

For how many rounds can Carol do this? Solving for $s$ in $\sum_{i=2/b}^{s} 2^{bi}/2\geq \Total$ implies that $s \leq \lg(\Total)/b + O(1)$. To this, we add the number of non-blocked phases Alice endures before terminating; let $X$ denote the random variable for this number of rounds.  Then, $E[X] \leq (1-e^{-2}) + 2\,e^{-2}(1-e^{-2}) + 3\,e^{-4}(1-e^{-2}) + ... = \sum_{i=1}^{\infty} i\,e^{2(1-i)}(1-e^{-2}) = (1-e^{-2})e^2 \sum_{i=1}^{\infty} i(e^{-2})^i$  by Lemmas~\ref{lemma:nonjamming1} and \ref{lemma:nonjamming2}.  Therefore, $E[X]\leq 1/(1-e^{-2}) = O(1)$ which means Carol can delay the termination by an additional $O(1)$ rounds. 

For each of the $\lg(\Total)/b + O(1)$ rounds, Alice experiences a delay of $2^{ci} + 2^{bi} = O(2^{ci})$ slots. This means that Alice terminates within an expected  $O(2^{(c/b)\lg(\Total) + O(c)}) = O(\Total^{c/b})$ slots. This translates to $O(\Total^2)$ and $O(\Total^{\varphi})$ for Cases 1 and 2, respectively.  $\qed$
\vspace{8pt}

\noindent Figure~\ref{fig:final-alg-set} provides an updated \Alg~with a setting of the constants $a=\varphi-1,b=1,$ and $c=\varphi$ such that we obtain the guarantees of Theorem~\ref{thm:main1} for Case 2.

\begin{figure}[t]
\begin{center}
\fbox{\parbox[t]{3.4in}{\noindent{}\Alg~for round $i\geq 2$\vspace{5pt}

\noindent
\emph{Send Phase:} For each of the $2^{\varphi i}$ slots do\vspace{-0pt}
\begin{itemize}
\item If Alice has sent in less than $2^{(\varphi-1)i+2}$ slots in this phase, she sends $m$ with probability $2/2^{i}$\vspace{2pt}
\item If Bob has listened in less than $2^{i+2}$ slots in this phase, he listens with probability $2/2^{(\varphi-1)i}$\vspace{-0pt}
\end{itemize}
If Bob received $m$, then he terminates\\ \vspace{-4pt}

\noindent
\emph{Nack Phase:} For each of the $2^{i}$ slots do\vspace{-2pt}
\begin{itemize}
\item Bob sends a \texttt{nack} message\vspace{2pt}
\item If Alice has listened in less than $2^{\varphi i+2}$ slots in this phase, she listens with probability $4/2^{i}$\vspace{2pt}
\end{itemize}

\noindent
If Alice listened to a clear slot, then she terminates.

} 

} 
\end{center}
\vspace{-14pt}
\caption{Pseudocode for \Alg~with the constants $a=\varphi-1$, $b=1$, $c=\varphi$.}\label{fig:final-alg-set}
\vspace{-10pt}\end{figure}


\section{Tolerating a Reactive Adversary}\label{section:reactive_adv} 

Consider a reactive adversary Carol who can detect channel activity for free~\textemdash~that is, distinguish between when the channel is clear and when it is busy~\textemdash~and then jam. This ability to reactively jam is possible in WSNs~\cite{wilhelm:reactive}.

Carol can now detect that $m$ is being sent in the Send Phase and jam it without fail. To address this powerful adversary, we consider the case where non-critical data, $m'$, is sent over the channel by other participants in addition to Alice and Bob.  Carol can detect the traffic; however, she cannot discern whether it is $m$ or $m'$ without listening to a portion of the communication. For example, in practice, she may need to listen to a part of  the packet header in order to make such a determination.

In a slot where channel activity is detected, if Carol listens for a portion of the message, we assume she incurs a cost. Therefore, the cost to Carol is proportional to (1) the number of messages to which she listens, and (2) the number of slots in which she jams. Importantly, in the presence of $m'$,  Carol's ability to detect is unhelpful since $m'$ provides ``camouflage'' for $m$.  

As an example, assume that {\it all}  slots in the Send Phase are used either by Alice to send $m$ (as per our algorithm) or another party, Dave, whose transmissions of $m'$ do not interest Carol. Assume that Dave's transmissions are in the majority. 

In this situation, detecting channel activity does not help Carol decide on whether to jam; all slots are used. Regardless of how she decides to act, Carol can do no better than picking slots independent of whether she detects channel activity. In other words, channel activity is no longer useful in informing Carol's decisions about whether to jam.

But assuming {\it all} slots are active is problematic.  How is this guaranteed or coordinated? We relax our assumption that all slots are used. Instead, we assume that other network traffic occurs such that Carol will always detect traffic on at least a {\it small} constant fraction of slots in the Send Phase. This traffic is assumed to be random and independent of Alice sending $m$.

Upon detecting traffic, Carol may listen to a portion of the message to discover if it is $m$ or $m'$ and then decide on whether to jam, all at a cost of $1$. However, this is roughly as expensive~\textemdash~it is the same order of magnitude~\textemdash~as simply jamming outright (also a cost of $1$ in our model).  

In practice, such situations can arise where communication occurs between many participants, or via several distributed applications, or internally between respective co-located networks. An attacker may wish to selectively target only one component of the system in order to conserve energy and reduce the chances of its malicious activity being detected.

As with any two messages,  if $m$ and $m'$ are sent over the channel in the same slot, the two messages collide and Bob receives neither.  Define a slot as \defn{active} if either $m$ or $m'$ is sent in that slot. For this result only, redefine a \defn{blocked Send Phase} as  one where Carol listens to or jams more than a $1/3$-fraction of the {\it active} slots in the Send Phase; otherwise, it is \defn{non-blocked}.  We  provide a result analogous to Lemma~\ref{lemma:nonjamming1}.\vspace{\nval}

\begin{lemma}\label{lemma:nonjamming1_reactive} Let Carol be an adaptive and reactive adversary. Then, in a non-blocked Send Phase, the probability that Bob does not receive $m$ from Alice is at most $e^{-2}$.\vspace{\nval}
\end{lemma}
\begin{proof}
Let $x=2^{\varphi i}$ be the   number of   slots  in the Send  Phase. Consider the set of slots used
by all participants (such as Dave) other than Alice.  We assume these participants pick their slots at random
to send, so that for any slot the probability is 2/3 that the slot is chosen by at least one of
them. Let $p_A$ and $p_B$ be the probabilities of sending and listening by Alice and Bob in the Send Phase, respectively.

Since we assume  these messages $m'$ are sent independently at random, then  Chernoff bounds
imply that, with high probability (i.e., $1-1/x^{c'}$ for constants $c', \epsilon$ and sufficiently large $x$) the number of slots  $y$ during which $m'$ is sent
is  greater than  $(2x/3)(1 - \epsilon)$  where
$x$ is the total number of slots in a phase. 

In the same way, assume the number of slots in which Alice sends is at least $a'=(1-\delta)xp_A = (1-\delta)2^{(\varphi-1)i+1}$
with probability $1-1/x^{c''}$ for a constant $\delta,c''$ and sufficiently large $x$. The number of active slots  is clearly at least $y$.

By definition of a non-blocked Send Phase, Carol listens to or blocks at most 
$x/3$ active slots in the Send Phase. As Carol has no information about the source of a
message sent in an active slot until  she  listens to it, her choice is independent of the source
of the message.

Given a slot in which Alice sends,  there is  at least a $1-(x/3)/y $ chance it will
 not be listened to or jammed by Carol. The probability that this slot will not be used by another participant is
 $1/3$ and the probability that Bob will listen to the slot is $p_B$. Hence the probability of a successful  transmission
from Alice to Bob in a slot for which Alice sends is at least: 
\begin{eqnarray*} 
\left(1-\frac{x}{3y}\right) \left(\frac{1}{3}\right) p_B &=& \left(1-\frac{1}{2(1-\epsilon)} \right)  \left(\frac{1}{3}\right) p_B \\
&\geq& \left(\frac{1}{3} - \frac{(1+\delta)}{6}\right)p_B\\
& \geq & \left(\frac{1}{12}\right)p_B
\end{eqnarray*} 
\noindent for sufficiently large $x$ (that reduces the size of $\delta$) when $y>(1-\epsilon)(2x/3)$.  

The probability that all messages that Alice sends fail to be delivered is at most $(1- p_B/12)^{a'} + 2/x^{c''}$ where the
last term is the probability of the bad event that $y$ or $a'$ is small and $c''>0$ is a constant.  Redefine $p_B=24/((1-\delta)2^{(\varphi-1)i})$ where the value $24$ is set to off-set the additive $2/x^{c''}$ term. Note that this constant factor increase in the listening probability does not change our asymptotic results and our analysis  in Section~\ref{section:analysis}  proceeds  almost identically. Therefore, we then have  $(1- p_B/12)^{a'} + 2/x^{c''} \leq e^{-2}$.\qed
\end{proof}

\noindent{}The \Alg~can be modified so that the initial value of $i$ is large enough to render the error arising from the use of Chernoff bounds sufficiently small; we omit these details. 

The required level of channel traffic detected by Carol is flexible. Different values can be accommodated if the parties' probabilities for sending and listening are modified appropriately in \Alg; our results hold asymptotically.   We emphasize that further revision of the arguments presented in Section~\ref{section:analysis} are minor.  Lemmas~\ref{lemma:case1},\ref{lemma:nonjamming2}, and~\ref{lemma:costs} do not require modification.  Carol cannot decide to block only when Alice is listening since detecting when a node is listening is assumed to be impossible under our model. Alternately, Carol cannot silence a \texttt{nack} through (reactive) jamming since this is still interpreted as a retransmission request.  Using Lemma~\ref{lemma:nonjamming1_reactive}, Theorem~\ref{thm:main1} follows as before with the same asymptotic guarantees.

Finally, we note that the conclusion of our argument aligns with claims put forth in empirical results on reactive jamming; that is, such behavior does not necessarily result in a more energy-efficient attack because the adversary must still be listening to the channel for broadcasts prior to committing itself to their disruption~\cite{xu:feasibility}.


\section{Resource-Competitive Local Broadcast}\label{section:multiple} 
 
We present  a resource-competitive algorithm, \AlgN, that allows Alice to send a message $m$ to $n$ neighboring \defn{receivers} within her transmission range; that is, our algorithm allows Alice to perform  a local broadcast. We assume that the quantity $\ln n$ is known to Alice; however, our results likely still hold if a constant-factor approximation is known.

Throughout this section, we only consider the equivalent of Case 2 in this section; that is, where Carol may spoof any receiver(s) and can jam their transmissions. However, equivalent results trivially hold for the simpler Case 1. \AlgN~preserves much of the same guarantees as \Alg, and is fair to all parties up to a polylogarithmic factor. 

Our pseudocode is given in Figure~\ref{fig:pseudocode_nparty}. The probabilities for sending and listening are modified from $\Alg$. Note that \texttt{nack} messages from multiple receivers can (and likely will) collide in the Nack Phase. This is fine since such a collision is due to either jamming or multiple  receivers requesting a retransmission; in either case, Alice will correctly resend. 

If Carol jams, then we assume she chooses the subset of the receivers (which can be the entire set of receivers) whose members do not receive $m$. In the literature, such an adversary is referred to as $n$-uniform~\cite{richa:jamming2}. The cost for jamming is still $1$ per slot regardless of how many receivers are jammed. We maintain the definition in Section~\ref{sec:less-naive} of a blocked Send Phase (or Nack Phase): more than half the slots are jammed.

In this $n$-party setting, the arguments regarding the Send Phase are in need of minor revision, and we repeat these for completeness. We use nearly-identical definition of a failed round: any round in which either Alice or any receiver exceeds her/his respective expected cost by a factor of $2$ or more in either the Send or Nack phase; otherwise, it is \defn{non-failed round}. In the multi-receiver case, we can bound the probability of a failed round as a function of $n$. 

\begin{lemma}\label{lemma:deviation-n}
For any round $i\geq \lg(4\ln n)$ the probability of a failed round is at most $1/n$.
\end{lemma}
\begin{proof}
In the Send Phase of round $i$, let the binary random variable $X_j=1$  if Alice sends $m$ in slot $j$; otherwise, $X_j=0$. Letting $X=\sum_{j=1}^{2^{\varphi i}} X_j$, then:
$$E[X] = \sum_{j=1}^{2^{\varphi  i}} E[X_j] = \sum_{j=1}^{2^{\varphi i}} (2\ln n/2^{i}) = 2^{(\varphi-1)i+1}\ln n$$

\noindent by linearity of expectation.  A similar, calculation for the Nack Phase yields an expected cost of $4\ln n$. Therefore, over epoch $i$, Alice's expected cost is $2^{(\varphi-1)i} + 4\ln n$.  By Theorem~\ref{thm:Chernoff}, the probability that she reaches her cutoff is at most:
$$\exp{(-2^{(\varphi -1 )i+1}\ln n/3)} \leq \exp{(-\Theta(\ln^{\varphi} n))}$$ 

\noindent since $i\geq  \lg(4\ln n)$. In the Send Phase of round $i$,  let the binary random variable $Y_j=$  if receiver $R$ listens in slot $j$; otherwise, $Y_j=0$. Letting $Y=\sum_{j=1}^{2^{\varphi  i}} Y_j$, then:
$$E[Y] = \sum_{j=1}^{2^{\varphi i}} E[Y_j] = \sum_{j=1}^{2^{\varphi i}} (2/2^{(\varphi-1)i}) = 2^{i}$$ 

\noindent by linearity of expectation. There is an additional $2^{i}$ cost associated with the Nack Phase. By Theorem~\ref{thm:Chernoff}, the probability that $R$  reaches its cutoff is at most:
$$\exp{(-2^{i+1}/3)} \leq \exp{(-8/3)\ln n} \leq n^{-8/3}$$
\noindent Taking a union bound implies that the probability that any receiver reaches its respective cutoff is at most $n^{-5/3}$. Finally, adding Alice's error from the above calculations yields the result.\qed
\end{proof}

\noindent Next, we prove the analogue to Lemmas~\ref{lemma:nonjamming1} and~\ref{lemma:case1} in the multi-receiver case.

\begin{lemma}\label{lemma:nonjamming1-n} Consider a blocked Send Phase in a non-failed round. The probability that any receiver does not receive the message from Alice  is at most $1/n$.
\end{lemma}\vspace{\pval}
\begin{proof}
Let $s=2^{\varphi i}$ be the number of slots in the Send Phase. Let $p_A$ be the probability that Alice sends in a particular slot. Let $p_R$ be the probability that a fixed receiver listens in a particular slot. Let $X_j=1 $ if the message is not delivered from Alice to the receiver in the $j^{th}$ slot. Then: 
\begin{eqnarray*}
&& Pr[\mbox{$m$ is not delivered in the Send Phase}]\\
&=&Pr[X_1 X_2\cdots{}X_s=1]\\
&=& Pr[X_s=1~|~ X_1~X_2\cdots{}X_{s-1}=1]\cdot{}\prod_{i=1}^{s-1}Pr[X_i=1]
 \end{eqnarray*}
  Let $q_j=1$ if Carol does not jam in slot $j$; otherwise, let $q_j=0$.  The value of $q_j$ can be selected arbitrarily by Carol. Then:
\begin{eqnarray*}  
Pr[X_i=1~|~ X_1 X_2\cdots{}X_{i-1}=1] &=& 1-p_Ap_Rp_Rq_j 
\end{eqnarray*}
\noindent Substituting for each conditional probability, we have:
\begin{eqnarray*}  
Pr[X_1X_2\cdots{}X_s=1] &=&(1-p_Ap_Rq_1)\cdots{}(1-p_Ap_Rq_s)\\
&=& \prod_{j=1}^{s}(1-p_Ap_Rq_j)\\
&\leq & e^{-p_Ap_R \sum_{j=1}^s q_j}\\
&\leq & e^{-(4\ln n/2^{\varphi i})(s/2)} \\
& = & e^{-2\ln n} \\
& = & n^{-2}
\end{eqnarray*}
\noindent Taking a union bound over all $n$ receivers yields the result. \qed
\end{proof}

\begin{lemma}\label{lemma:case1-n}
Assume there are no blocked Send Phases and no blocked Nack Phases. With high probability, the cost to Alice and each receiver is $O(\ln^{\varphi} n)$ and $O(\ln n)$, respectively.
\end{lemma}
\begin{proof}
For ease of exposition, let $\sigma  = \lg(4\ln n)$. We compute the expected cost for both non-failed and failed rounds for $i\geq \sigma$. Using Lemmas~\ref{lemma:deviation-n} and~\ref{lemma:nonjamming1-n}, the expected cost to Alice is at most:
\begin{eqnarray*}
\hspace{1cm}&& \sum_{i=\sigma}^{\infty}  (1/n)^{(i-\sigma)}\,O(2^{(\varphi-1)i}\ln n) = O(\ln^{\varphi} n) 
\end{eqnarray*}
\noindent Similarly, the expected cost to each receiver $R$ is at most:
\begin{eqnarray*}
\hspace{1cm}&& \sum_{i=\sigma}^{\infty}  (1/n)^{(i-\sigma)} O(2^{i})  = O(\ln n) 
 \end{eqnarray*} 
 
\noindent which completes the proof. Using a Chernoff bound (Theorem~\ref{lemma:costs}) provides the guarantee with high probability. \qed
\end{proof}

\begin{lemma}\label{lemma:nonjamming2-n} Assume that all receivers have received $m$ by non-failed round $i$ and that round $i$
is a non-blocked Nack Phase. Then the probability of a Nack Failure is less than $1/n^{2}$.
\end{lemma}\vspace{\pval}
\begin{proof}
Let $s=2^{i}$ be the number of slots in the Nack Phase and let $p=4\ln n/2^{i}$ be the probability that Alice listens in a slot.
For slot $j$, define $X_j$ such that $X_{j}=1$ if Alice does not terminate.  Then, $Pr[$ Alice retransmits $m$ in round $i+1] = $ $Pr[X_1X_2\cdots{}X_s=1]$. Let $q_j=1$ if Carol does not jam  in slot $j$; otherwise, let $q_j=0$. The $q_j$ values are determined arbitrarily by Carol. Since Alice terminates when she listens and hears a clear slot, then $Pr[X_j=1]=(1-pq_j)$. Therefore: 
\begin{eqnarray*}
Pr[X_1X_2\cdots{}X_s=1] & \leq & e^{-p\sum_{j=1}^{s}q_j} \\
& \leq  & e^{-2\ln n} \\
& \leq  & n^{-2}
\end{eqnarray*}
\noindent which completes the proof. \qed
\end{proof}

\noindent The analogue to Lemma~\ref{lemma:costs} is mostly unchanged, although we first establish the following result:

\begin{lemma}\label{lemma:whp-n}
With high probability, the cost of a round $i$ is $O(2^{(\varphi-1)i}\ln n)$ and $O(2^{i})$ for Alice and each receiver, respectively.
\end{lemma}
\begin{proof}
For the Send Phase, define a binary random variable $X_j=1$ if Alice sends in slot $j$; otherwise, $X_j=0$. Setting $X=\sum_{j=1}^{2^{\varphi i}} X_j$, we have:
\begin{eqnarray*}
E[X] & = & \sum_{j=1}^{2^{\varphi i}}  E[X_j ] =  \sum_{j=1}^{2^{\varphi i}}  2\ln n/2^{i} = O(2^{(\varphi - 1)i}\ln n).
\end{eqnarray*}
\noindent By Theorem~\ref{thm:Chernoff}, the cost is within a small constant factor with high probability. Similarly, in the Nack Phase, let the indicator random variable $Y_j=1$ if Alice listens in slot $j$.  Setting  $Y=\sum_{j=1}^{2^{\varphi i}} Y_j$, we have:
\begin{eqnarray*}
E[Y] & = & \sum_{j=1}^{2^{i}}  E[Y_j ] =  \sum_{j=1}^{2^{i}}  4\ln n/2^{i} = 4\ln n.
\end{eqnarray*}
\noindent By Theorem~\ref{thm:Chernoff}, we can bound this to within a small constant factor with high probability. Therefore, with high probability, Alice's cost is $O(2^{(\varphi - 1)i}\ln n) + O(\ln n) = O(2^{(\varphi - 1)i}\ln n)$.

Consider a receiver $R$. Let $X'_j=1$ if $R$ listens in slot $j$ of the Send Phase; otherwise, $X'_j=0$. Setting $X'=\sum_{j=1}^{2^{\varphi i}} X'_j$, we have:
\begin{eqnarray*}
E[X'] & = & \sum_{j=1}^{2^{\varphi i}}  E[X'_j ] =  \sum_{j=1}^{2^{\varphi i}}  2/2^{(\varphi-1)i} = O(2^{i}) = O(\ln n)
\end{eqnarray*}
\noindent where the last equality follows from $i\geq \lg(4\ln n)$. By Theorem~\ref{thm:Chernoff}, we can bound this to within a small constant factor with high probability. Taking a union bound over all $n$ receivers establishes the result. \qed
\end{proof}

\begin{lemma}\label{lemma:costs-n}
Assume there is at least one blocked Send Phase. With high probability, the cost to Alice and each receiver is $O(\Total^{\varphi-1})$.\end{lemma}
\begin{proof}
\noindent{}Let $i\geq{}\lg(4\ln n)$ be the last blocked Send Phase. Let $j\geq{}i$ be the last blocked Nack Phase; if no such blocked Nack Phase exists, then assume $j=0$.  Carol may choose $i$ and $j$ depending on the history of the execution. With this notation in place, we can bound the cost to Carol, Alice, and the receivers.  \smallskip 

\noindent{\bf Carol:} The total cost to Carol is $\Total = \Omega(2^{\varphi i} + 2^{j})$. \smallskip


\noindent{\bf Alice:}  Using Lemmas~\ref{lemma:deviation-n},~\ref{lemma:nonjamming1-n}, and~\ref{lemma:whp-n}, with high probability the cost to Alice prior to $m$ being delivered is at most:
\begin{eqnarray*}
&& O(2^{(\varphi-1)i}\ln n) + \sum_{k=1}^{\infty}  (2/n)^{k-1}  \,O(2^{(\varphi-1)(i+k)}\ln n) \\
&=&  O(2^{(\varphi-1) i}\ln n). 
\end{eqnarray*}

Next, using Lemmas~\ref{lemma:deviation-n},~\ref{lemma:nonjamming2-n}, and~\ref{lemma:whp-n}, we calculate the  cost to Alice after delivery of $m$; this addresses blocked Nack Phases. By assumption, the last blocked Nack Phase occurs in round $j$ and, therefore, with high probability Alice's expected cost is at most:
\begin{eqnarray*}
&& O(2^{(\varphi-1)j}\ln n) + \sum_{k=1}^{\infty} (2/n)^{k-1} O(2^{(\varphi-1)(j+k)}\ln n)\\
&=&O(2^{(\varphi-1)j}\ln n). 
\end{eqnarray*}
\noindent Therefore, with high probability, the cost to Alice is  $O(\Total^{\varphi-1})$. \smallskip 

\noindent{\bf Receivers:} We analyze the cost to each receiver up until $m$ is received. Note that we assume Carol does not have any cost for blocked Nack Phases up until this point since, as discussed in Section~\ref{sec:description}, there is no benefit to causing a Nack Failure via jamming. 

Using  Lemmas~\ref{lemma:deviation-n},~\ref{lemma:nonjamming1-n}, and~\ref{lemma:whp-n}, prior to receiving $m$, with high probability each receiver has cost at most:
\begin{eqnarray*}
&& O(2^{i}) + \sum_{k=1}^{\infty}  (2/n)^{k-1}O(2^{i+k}) \\
&=& O(2^{i} + \ln n) \\
&=& O(2^{i}).
\end{eqnarray*}
\noindent  Thus, with high probability, the cost to each receiver is $O(\Total^{1/\varphi}) = O(\Total^{\varphi-1})$. $\qed$
\end{proof}

\begin{lemma}\label{lemma:latency-n}
With high probability, Alice and all correct receivers terminate \AlgN~in $O(\Total^{\varphi} + \ln^{\varphi}n)$ slots. 
\end{lemma} 
\begin{proof}
We analyze the time it takes for Alice to terminate since she does so only after all receivers terminate.  As with the analysis of \Alg, we pessimistically assume a blocked phase prevents Alice from terminating in the current round. Furthermore, Carol will only perform her jamming to block Nack Phases since this maximizes the rounds that she keeps Alice alive. 

Solving $\sum_{i=\lg(4\ln n)}^s 2^i/2 \geq \Total$ yields that the number of rounds that Carol can block is at most $s \leq \lg(\Total) + O(1)$. By the end of the very next round, all receivers are guaranteed with high probability to receive the message by Lemma~\ref{lemma:nonjamming1-n},  and Alice is guaranteed with high probability to terminate by Lemma~\ref{lemma:nonjamming2-n}. Therefore, with high probability, Alice terminates after $O(2^{\varphi (\lg \Total + O(1))} + 2^{\varphi \lg(4\ln n)}) = O(\Total^{\varphi} + \ln^{\varphi}n)$.
\end{proof}

\noindent The following theorem on the resource-competitive properties of~\AlgN~follows directly from Lemmas~\ref{lemma:case1-n},~\ref{lemma:costs-n}, and~\ref{lemma:latency-n}

\begin{theorem}
Let Carol be an adaptive adversary that jams for $\Total$ slots. \AlgN~guarantees $m$ is received by all receivers and all parties terminate. With high probability, Alice has a cost of $O(\Total^{\varphi-1} + \ln^{\varphi}n)$, each receiver has a cost of $O(\Total^{\varphi-1} + \ln n)$, and all parties terminate within $O(\Total^{\varphi} + \ln^{\varphi}n)$ slots.
\end{theorem}

\begin{figure}[t]
\begin{center}
\fbox{\parbox[t]{3.4in}{
\noindent{}\mbox{ \AlgN~for $i \geq{}  \lg(4\ln{n}) $ }\vspace{3pt}

\noindent
{\it Send Phase:} For each of the  $2^{\varphi{}i}$ slots do\vspace{-0pt}
\begin{itemize}
\item If Alice has sent in less than $2^{\varphi i + 2}\ln n$ slots in this phase, she sends $m$ with probability $\frac{2\ln{n}}{2^{i}}$.\vspace{-0pt}
\item If this receiver has listened in less than $2^{i+2}$ slots in this phase, it listens with probability $\frac{2}{2^{(\varphi-1)i}}$.\vspace{-0pt}
\end{itemize}

\noindent{}Any receiver that obtains $m$ terminates.\vspace{2pt}

\noindent
{\it Nack Phase:} For each of the $2^{i}$ slots do\vspace{-0pt}
\begin{itemize}
\item Each receiver that has not terminated sends \texttt{nack}.\vspace{-0pt}
\item If Alice has listened in less than  $2^{\varphi i + 2}\ln n$ slots in this phase, she listens with probability $\frac{4\ln{n}}{2^{i}}$.\vspace{-0pt}
\end{itemize}

\noindent{}If Alice listened to a clear slot, then she terminates.
} 

} 
\end{center}
\vspace{-15pt}\caption{Pseudocode for \AlgN.}\label{fig:pseudocode_nparty}\vspace{-10pt}
\end{figure}


\section{Discussion}\label{section:discussion}\vspace{\nval}

In this section, we provide some follow-up discussion regarding the practical aspects of wireless LPNs and how our abstract network model is motivated. \smallskip

\noindent{\bf Sending and Listening Costs in Practice:} Wireless network cards typically offer states such as {\it  sleep, receive (or listen)} and {\it transmit (or send)}. While the sleep state requires negligible power, the cost of the send  and listen states are roughly equivalent and dominate the operating cost of a device.  For example, the send and listen costs for the popular Telos motes are 38mW and 35mW, respectively and the sleep state cost is $15 \mu{}$W~\cite{polastre:telos}; therefore, the cost of the send/listen state  is more than a factor  of $2000$ greater and the sleep state cost is negligible. In the context of our work, when a party is not active, we can assume it is in the energy-efficient sleep state.

Disruption may not require jamming an entire slot, and so jamming a slot can be less costly than sending/listening. However, we can assume a small $m$ (or $m$ is broken into small packets) such that jamming and sending costs are within a constant factor of each other. \smallskip 

\noindent{\bf Slots:}  A time-slotted network corresponds to a time division multiple access (TDMA)-like medium access control (MAC) protocol; that is, a time-slotted network. Two examples are the popular IEEE 802.11 family of specifications, and the well-known LEACH~\cite{heinzelman:energy}.  For simplicity, a {\it global} broadcast schedule is assumed; however, this is likely avoidable if nodes maintain multiple schedules as with S-MAC~\cite{ye:energymac}. Even then, global scheduling has been demonstrated by experimental work in~\cite{li:energy} and secure synchronization has been shown~\cite{ganeriwal:secure}.

Clear channel assessment (CCA), which subsumes carrier sensing, is a common feature on devices for detecting activity on the channel \cite{ramachandran:clear}; this is considered practical under IEEE 802.11~\cite{deng:new}. Collisions are usually only detectable by the receiver~\cite{wood:denial}. When a collision occurs,   a correct node discards any received data. We assume that the absence of channel activity cannot be forged by the adversary; this aligns with the empirical work by Niculescu~\cite{niculescu:interference} who shows that channel interference increases linearly with the combined rate of the sources. Finally, we also note that several theoretical models feature collision detection (see~\cite{alistarh:securing,richa:jamming2,gilbert:malicious,awerbuch:jamming,koo2}).\\

\noindent{\bf On Reactive Adversaries:} CCA is performed via the radio chip using the {\it received signal strength indicator} (RSSI)~\cite{srinivasan:rssi}. If the RSSI value is below a clear channel threshold, then the channel is assumed to be clear~\cite{wildpackets}. Such detection consumes on the order of $10^{-6}$\hspace{2pt}W which is three orders of magnitude smaller than the send/listen costs; therefore, a reactive Carol can detect activity (but not message content) at essentially zero cost. However, listening to even a small portion of a message costs on the order of milliwatts and our argument from Section~\ref{section:reactive_adv} now applies. \smallskip

\noindent{\bf Cryptographic Authentication:} We assume that the message $m$ from Alice can be authenticated. Therefore, Carol cannot spoof Alice. Several results show how light-weight cryptographic authentication can be implemented in sensor networks~\cite{liu:multi,watro:tinypk,law:survey,karlof:tinysec,walters:security}. 

What about having a shared secret for Alice and Bob that helps coordinate their communication? This is outside the scope of this work, but we remark that the adversary may capture a limited number of parties (such as Bob).  These parties are said to suffer a Byzantine fault and are controlled by the adversary~\cite{wood:denial,walters:security}. Given this attack, we emphasize that, while we assume authentication on Alice's side, attempts to share a secret
send/listen schedule between Alice and Bob allows Carol to manipulate parties in ways that appear difficult to overcome.  


\section{Conclusion and Future Work}\label{section:conc}

In this paper, we provided resource-competitive algorithms for mitigating  jamming attacks in wireless LPNs.  We see that the golden ratio arises naturally from our analysis, and its appearance in this adversarial setting is interesting. Notably, a later result in~\cite{gilbert:near} demonstrates that our result for Alice and Bob is asymptotically tight in that $\Omega(\Total^{\varphi-1})$ expected cost is necessary.

Future work includes pursuing resource-competitive algorithms to mitigate jamming in multihop networks. Additionally,  the application of these results to problems of consensus and leader election may prove fruitful. Finally, investigations into more sophisticated communication models, such as the signal-to-interference-plus-noise (SINR) model, may be of interest.


\end{document}